\documentclass[a4paper,twocolumn,11pt,accepted=2025-04-23]{quantumarticle}
\pdfoutput=1

\usepackage{enumerate,appendix}
\usepackage{amsmath, amsthm, amssymb,commath}
\usepackage{color, calc,graphicx}
\usepackage[usenames,dvipsnames,svgnames,table,cmyk,hyperref]{xcolor}
\usepackage[colorlinks]{hyperref}
\usepackage{optidef}
\hypersetup{
	colorlinks = true,
	urlcolor = {blue},
	citecolor = {blue},
	linkcolor= {blue}
}

\usepackage{graphicx}
\usepackage{amsmath}
\usepackage{latexsym}
\usepackage{bbm}
\usepackage{bm}

\usepackage{booktabs}
\usepackage{multirow}
\usepackage{subcaption}
\usepackage{dcolumn}
\usepackage{mathrsfs}

\usepackage{algorithm}
\usepackage{algcompatible}
\usepackage{enumitem}

\def\SU{\mathop{\rm SU}}
\newcommand{\frS}{\mathfrak{S}}

\def \be {\begin{equation}}
\def \ee {\end{equation}}

\newcommand{\Tr}{\mathrm{Tr}}

\def \MSE{\mathop{\rm MSE}\nolimits}
\def \polylog{\mathop{\rm polylog}\nolimits}

\def \mix{\mathop{\rm mix}\nolimits}

\def \cH{{\cal H}}

\def \sofc2{{\cal S}({\mathbb C}^2)}

\newtheorem{theorem}{Theorem}
\newtheorem{lemma}[theorem]{Lemma}
\newtheorem{proposition}[theorem]{Proposition}

\newtheorem{remark}[theorem]{Remark}

\def\Label#1{\label{#1}\ [\ \text{#1}\ ]\ }
\def\Label{\label}

\def\blambda{\bm{\lambda}}

\begin{document}

\title{Measuring quantum relative entropy with finite-size effect}

\author{Masahito Hayashi}
\address{School of Data Science, The Chinese University of Hong Kong,
Shenzhen, Longgang District, Shenzhen, 518172, China}
\address{International Quantum Academy, Futian District, Shenzhen 518048, China}
\address{Graduate School of Mathematics, Nagoya University, Furo-cho, Chikusa-ku, Nagoya, 464-8602, Japan}
\orcid{0000-0003-3104-1000}
\email{hmasahito@cuhk.edu.cn, masahito@math.nagoya-u.ac.jp}
\maketitle

\begin{abstract}
We study the estimation of relative entropy $D(\rho\|\sigma)$
when $\sigma$ is known.
We show that the Cram\'{e}r-Rao type bound equals the relative varentropy.
Our estimator attains the Cram\'{e}r-Rao type bound when 
the dimension $d$ is fixed.
It also achieves the sample complexity $O(d^2)$ when the dimension $d$ increases.
This sample complexity is optimal when $\sigma$ is the completely mixed state.
Also, it has time complexity $O(d^6 \polylog d)$.
Our proposed estimator unifiedly works under both settings.
\end{abstract}


\section{Introduction}
Distinguishing a given known state $\sigma$ from another state $\rho$ is a fundamental task in quantum information.
When the state $\rho$ is known as well as $\sigma$ 
and $n$ copies of the unknown state are prepared,
quantum Stein's lemma \cite{HP,ON} guarantees that 
the quantum relative entropy $D(\rho\|\sigma)
:= \Tr \rho (\log \rho -\log \sigma)$ gives
the limitation of the optimal discrimination power.
The measurement 
achieving the optimal performance depends only on the state $\sigma$ \cite{H-02}.
Further, quantum Sanov theorem \cite{Bjelakovic,Notzel} shows that 
the construction of the test 
achieving the optimal performance depends only on the state $\sigma$ and
the value of the quantum relative entropy $D(\rho\|\sigma)$.
That is, when we construct the test achieving the optimal performance,
the information of the quantum relative entropy $D(\rho\|\sigma)$ is needed
in addition to the information of $\sigma$. 
Although the papers \cite{AISW,WZ} studied the estimation of von Neumann entropy and 
R\'{e}nyi entropy, they did not study the estimation of the quantum relative entropy $D(\rho\|\sigma)$ with the known state $\sigma$.
Due to the non-commutativity between $\rho$ and $\sigma$,
the method in \cite{AISW,WZ} cannot be applied to the estimation of 
the quantum relative entropy $D(\rho\|\sigma)$ even when the state $\sigma$ is known.
In particular, the simple use of Schur sampling \cite{KW01,CM,CHM,OW,HM02a,HM02b}
does not work for the estimation of 
the quantum relative entropy $D(\rho\|\sigma)$
although the paper \cite{AISW} employs Schur sampling.
Therefore, this paper studies how to estimate 
the quantum relative entropy $D(\rho\|\sigma)$ with $n$ copies of the unknown state $\rho$
when $\sigma$ is known.

In fact, when we intend to prepare a specific state $\sigma$,
it is crucial to identify the difference between 
the known target state $\sigma$ and the unknown actual real state 
$\rho$.
In this situation, we often focus on the fidelity $F(\rho,\sigma)
:=\Tr |\sqrt{\rho}\sqrt{\sigma}|$ for this aim.
Many papers \cite{Ha09-2,MH15-7,PLM,ZH19,LZH} studied this problem
when $\sigma$ is a pure state.
However, no paper focuses on this topic 
when $\sigma$ is a mixed state.
Instead of the fidelity, one can evaluate the difference between
$\sigma$ and $\rho$ by using the quantum relative entropy $D(\rho\|\sigma)$
because of the bounding inequality $-2\log F(\rho,\sigma)\le 
D(\rho\|\sigma)$ \cite[Eq.(3.49)]{Springer}.
That is, our problem works for evaluating this difference.

The estimation of $D(\rho\|\sigma)$ with the unknown state $\rho$
is formulated as a special case of the estimation of a function output with nuisance parameters.
This general problem setting has been studied in 
various researches \cite{WSU,YCH,SYH,DGG,GC,TAD}.
The reference \cite[Section III]{TAD} addressed the estimation problem 
of $D(\rho\|\sigma)$ as a special case, and
showed that the estimation error is lower bounded by 
the relative varentropy
$V(\rho\|\sigma):= \Tr \rho ( 
(\log \rho -\log \sigma)-D(\rho\|\sigma))^2$, which appears
in the second-order theory of hypothesis testing \cite{TH,KLi}
while it did not show its attainability.
Under a general setting, the reference \cite[Section 4.3]{SYH} showed that
the optimal precision of the estimation  
is given by the Cram\'{e}r-Rao type bound $C$ with nuisance parameters, which equals the inverse of
the projected symmetric logarithmic derivative (SLD) Fisher information
in the one-parameter case.
The meaning of the Cram\'{e}r-Rao type bound $C$ is the following:
when $n$ copies of the unknown state $\rho$ is prepared
under a fixed model,
the mean square of our best estimator behaves as
$\frac{C}{n}+o(\frac{1}{n}) $.
When we have no nuisance parameter,
the recent paper \cite{NF} showed the existence of an estimator whose mean square error (MSE) equals
$\frac{C}{n}$ without additional term under special models.
However, such an estimator does not exist in general.

In this paper, as the first step, using the results for estimation with nuisance parameters in \cite[Section 4.3]{SYH}, 
we show that the Cram\'{e}r-Rao type bound $C$
of this problem equals
the relative varentropy
$V(\rho\|\sigma):= \Tr \rho ( 
(\log \rho -\log \sigma)-D(\rho\|\sigma))^2$, which was obtained as a lower bound in \cite{TAD}.
In practice, it is necessary to find an estimator to achieve the 
MSE $\frac{C}{n}+o(\frac{1}{n}) $. 
%even with a certain approximation.
In addition, it is also desired to clearly evaluate 
the difference between its mean square error and the Cram\'{e}r-Rao type bound $\frac{C}{n}$.
Since our problem is one-parameter estimation,
it is possible to achieve the above bound without the use of 
collective measurements across several quantum systems. 
However, this method requires an adaptive choice of measurement.
A typical method is the two-step estimator, consisting of a pre-estimation stage and a local optimal estimation stage.
When we have $n$ copies of the unknown state $\rho$,
theoretically, in order to achieve the optimal bound,
it is sufficient to allocate $\sqrt{n}$ copies to the 
pre-estimation stage \cite{Ha11-2,YCH}.
Although this method achieves the optimal mean square error 
in the order $\frac{1}{n}$, i.e., the first order,
it is not so clear how the mean square error is close to the leading term
$\frac{C}{n}$.
From the practical viewpoint, it is necessary to upper bound this difference by using the number $n$ of unknown copies.
Unfortunately, the above method has difficulty to evaluate 
the difference between its mean square error and 
$\frac{C}{n}$.
The papers \cite{Fujiwara,RPB} also proposed different adaptive methods for parameter estimation, but their methods have the same problem. 
In addition, the above adaptive approach does not work when the dimension $d$ of the quantum system 
increases in addition to the number $n$ of copies. 

In summary, 
to account for finite-size effects in our estimation,
%considering the finite size effect of our estimation, 
we need to find an estimator of the quantum relative entropy $D(\rho\|\sigma)$
to satisfy the following conditions.
It asymptotically attains the Cram\'{e}r-Rao type bound $V(\rho\|\sigma)$ in the sense of the order $\frac{1}{n}$ when the dimension $d$ is fixed.
The difference between its mean square error and the Cram\'{e}r-Rao type bound is evaluated in terms of
the number $n$ of copies.
Furthermore, the estimator performs well 
even when the dimension $d$ of the quantum system 
increases in addition to the number $n$ of the prepared copies. 

As the second step of this paper, to satisfy the above requirements,
we propose an alternative estimation method by using Schur transform \cite{BCH}.
Then, we evaluate its mean square error and 
compare it to the Cram\'{e}r-Rao type bound with a finite number $n$.
In addition, we evaluate the tail probability. 

As the third step, similar to the papers \cite{AISW,WZ},
we study the case when the dimension $d$ increases when $n$ increases.
Unlike the estimation of von Neumann entropy,
the range of the quantum relative entropy or
the relative varentropy is not bounded even with fixed dimension $d$.
Therefore, to study the sample complexity, 
we need to restrict the possible range of $\rho$ and $\sigma$.
We show that the sample complexity of our method is $O(d^2)$
under a certain natural restriction for this range.
Our method extensively improves the sample complexity over the simple application of the full tomography \cite{HHJ}.
In particular, the restriction of the possible range of $\sigma$
is crucial in this analysis.
We also show that our time complexity is $O(d^6\polylog d)$.

When $\sigma$ is fixed to the completely mixed state,
our problem is the same as the estimation of 
von Neumann entropy, which was studied in the papers \cite{AISW,WZ}.
Our result contains the case studied in the papers \cite{AISW,WZ}.
In this special case, the paper \cite{AISW} also showed that the lower bound of the sample complexity is $O(d^2)$ 
for the Empirical Young Diagram (EYD) algorithm, which is given as
the entropy of the diagonal state determined by the Young indices.
Although this case improves the sample complexity over the simple application of the full tomography, the improvement is incremental. 
The paper \cite{WZ} also showed a lower bound $O(d)$ 
for the sample complexity, 
without restricting our algorithm to EYD algorithm.
In this special case, we show that our measurement can be restricted to 
Schur sampling.

The remaining part of this paper is organized as follows.
First, Section \ref{Sec2} derives the Cram\'{e}r-Rao type bound
for the estimation of quantum relative entropy.
Section \ref{Sec3} presents our estimator for quantum relative entropy,
and its performance.
Then, this section shows that this estimator asymptotically attains
the Cram\'{e}r-Rao type bound when $d$ is fixed.
Section \ref{Sec3B} studies the sample complexity of our method
when $d$ increases.
Section \ref{Sec3C} studies the time complexity of our method
when $d$ increases.
Appendices are devoted to the proofs of the results presented in Section \ref{Sec3}.

\section{Cram\'{e}r-Rao type bound}\Label{Sec2}
To clarify the minimum mean square error 
with the first order $\frac{1}{n}$, 
we discuss the Cram\'{e}r-Rao type bound $C$ 
with nuisance parameters, which 
is mathematically formulated as a lower bound of MSE under the local unbiasedness condition.
The reference \cite[(3.24), (3.52)]{TAD} showed that
the varentropy $V(\rho_*\|\sigma)$ is a lower bound 
under the estimation of $D(\rho\|\sigma)$ at $\rho=\rho_*$.
That is, they showed that
\begin{align}
C \ge V(\rho_*\|\sigma),\Label{NBU}
\end{align}
and did not discuss whether the equality holds.

In the following, we show the equality in \eqref{NBU}.
For this aim, we review the Cram\'{e}r-Rao type bound $C$ for a general function 
$f(\rho)$ in the full parameter model under a $d$-dimensional system as follows.
%we assume that our Hilbert space is $d$-dimensional, and 
We choose a parameterization 
$\rho_\theta$ with $\theta =(\theta^1, \theta^2, \ldots, \theta^{d^2-1})
\in \mathbb{R}^{d^2-1}$ such that
the first parameter gives the value of the function to be estimated
and the other parameters are orthogonal to the first parameter at 
$\rho=\rho_*$.
That is, our parameter needs to satisfy the following conditions.
\begin{align}
\rho_{\theta_0}&=\rho_*, 
\Label{Con1}\\
\frac{\partial f(\rho_{\theta})}{\partial \theta^1}\Big|_{\theta=\theta_0}&=1, \Label{Con2}\\
\frac{\partial f(\rho_{\theta})}{\partial \theta^j}\Big|_{\theta=\theta_0}&=0\Label{Con3}
\end{align}
for $j=2, \ldots, d^2-1$, where 
$\theta_0:= (\theta^1_0 ,0,\ldots,0)$,
and 
$\theta^1_0:=f(\rho_*)$.

For the latter discussion, 
we introduce the symmetric logarithmic derivative (SLD) inner product.
Given two Hermitian matrices $X,Y$, we define 
$X \circ Y:=( XY+YX)/2$.
Then, we define the SLD inner product 
$\langle X,Y\rangle_{\rho}:=
\Tr (\rho \circ X) Y=\Tr \rho (X\circ Y)$.
We define the SLD norm $\|X\|_\rho:=
\sqrt{\langle X,X\rangle_{\rho}}$.
We choose the SLD $L_1, \ldots, L_{d^2-1}$ as
\begin{align}
\frac{\partial \rho_{\theta}}{\partial \theta^j}\Big|_{\theta=\theta_0}
=\rho_* \circ L_j 
\end{align}
for $j=1, \ldots, d^2-1$.
We also assume the following orthogonal condition;
\begin{align}
\langle L_1,L_j\rangle_{\rho_*}=
\Tr 
\frac{\partial \rho_{\theta}}{\partial \theta^j}\Big|_{\theta=\theta_0}
L_1 =0\Label{Con3-2}
\end{align}
for $j=2, \ldots, d^2-1$.
Then, we can calculate the Cram\'{e}r-Rao type bound, i.e.,
the minimum mean square error under the local unbiasedness condition
even under the full parameter model.

\begin{proposition}[\protect{\cite[Theorem 5.3]{SYH}}]\Label{PP1}
Assume the conditions \eqref{Con1}--\eqref{Con3-2}.
The Cram\'{e}r-Rao type bound for the estimation of $f(\rho_\theta)$
is $\langle L_1,L_1\rangle_{\rho_*}^{-1}$ at $\theta=\theta_0$.
\end{proposition}

To calculate the Cram\'{e}r-Rao type bound for $D(\rho\|\sigma)$ at $\rho=\rho_*$,
we choose the Hermitian matrix 
%$L_1:= ((\log \rho_*-\log \sigma) - \frac{1}{d}\Tr [\log \rho_*-\log \sigma]I)
\begin{align}
L^1:= ((\log \rho_*-\log \sigma) - D(\rho_*\|\sigma)I),\Label{CB9}
\end{align}
and $d^2-1$ linearly independent Hermitian traceless matrices
$X_1,X_2, \ldots X_{d^2-1}$ such that
\begin{align}
X_1 &= V(\rho_*\|\sigma)^{-1}  \rho_* \circ L^1
 \Label{Con4}\\
\Tr X_j L^1 &=0\Label{Con5}
\end{align}
for $j=2, \ldots, d^2-1$.
Then, we choose our 
parameterization
$\rho_\theta$ as 
\begin{align}
\rho_\theta:= \rho_*+ (\theta^1-\theta^1_0 )X_1+
\sum_{j=2}^{d^2-1}\theta^j X_j.
\end{align}
This parameterization satisfies the condition \eqref{Con1}.
Since 
$\frac{\partial \rho_{\theta}}{\partial \theta^j}\Big|_{\theta=\theta_0}
=X_j$, 
we have
\begin{align}
&\frac{\partial D(\rho_{\theta}\|\sigma)}{\partial \theta^1}
\Big|_{\theta=\theta_0}\notag\\
=&
\Tr 
\frac{\partial \rho_\theta }{\partial \theta^1}\Big|_{\theta=\theta_0}
(\log \rho_*-\log \sigma)\notag\\
&+
\Tr \rho_* 
\frac{\partial \log \rho_\theta }{\partial \theta^1}\Big|_{\theta=\theta_0} \notag\\
\stackrel{(a)}{=} &
\Tr 
X_1
(\log \rho_*-\log \sigma)+
\Tr X_1\notag\\
\stackrel{(b)}{=} &
\Tr X_1 L^1
+\Tr X_1 D(\rho_*\| \sigma) I
+\Tr X_1
\stackrel{(c)}{=}1 \Label{ConT}.
\end{align}
Here, $(b)$ and $(c)$ follow from \eqref{CB9} and
the pair of \eqref{Con4} and the relation $\langle L^1,L^1\rangle_{\rho_*}= V(\rho_*\|\sigma)$, respectively.
In addition,
$(a)$ can be shown as follows.
\begin{align}
&\Tr \rho_{\theta_0}
\frac{\partial \log \rho_\theta }{\partial \theta^1}\Big|_{\theta=\theta_0} \notag\\
\stackrel{(a)}{=}&
\Tr \int_{0}^1 \rho_{\theta_0}^{s}
\frac{\partial \log \rho_\theta }{\partial \theta^1}\Big|_{\theta=\theta_0} 
\rho_{\theta_0}^{1-s} ds
\notag\\
=&\Tr \frac{\partial \rho_\theta }{\partial \theta^1}\Big|_{\theta=\theta_0} 
=\Tr X_1\Label{ConX},
\end{align}
where $(a)$ follows from \cite[Eq. (6.32)]{Springer}.

Using \eqref{Con5} and a similar discussion, we have
\begin{align}
\frac{\partial D(\rho_{\theta}\|\sigma)}{\partial \theta^j}
\Big|_{\theta=\theta_0}
=&
\Tr X_j
(\log \rho_*-\log \sigma)+
\Tr X_j\notag\\
=&0\Label{ConY}.
\end{align}
Then, \eqref{ConX} and \eqref{ConY}
imply the conditions \eqref{Con2} and \eqref{Con3}.

Due to \eqref{Con4} and \eqref{Con5}, we have
$L_1=V(\rho_*\|\sigma)^{-1} L^1$ and 
\begin{align}
\Tr \frac{\partial \rho_\theta}{\partial \theta^j} L_1 =
\Tr X_j L_1=0 \hbox{ for }j=2, \ldots, d^2-1,
\end{align}
which implies the condition \eqref{Con3-2}.
Therefore, 
since all the conditions \eqref{Con1} -- \eqref{Con3-2} hold,
due to Proposition \ref{PP1},
the Cram\'{e}r-Rao type bound for the estimation of 
$D(\rho \|\sigma)$
is $\langle L_1,L_1\rangle_{\rho_*}^{-1}= 
V(\rho_*\|\sigma)
$ at $\theta=\theta_0$.
That is,
the minimum mean square error under the local unbiasedness condition
is $ V(\rho_*\|\sigma)$.
That is,
the Cram\'{e}r-Rao type bound for the estimation of 
$D(\rho \|\sigma)$
is the relative varentropy $ V(\rho_*\|\sigma)$.
Thus, the quality in \eqref{NBU} holds.

\section{Estimator based on Schur transform}\Label{Sec3}
Similar to the papers \cite{H-01,H-02,KW01,CM,CHM,OW,HM02a,HM02b,Notzel,H-24,H-q-text},
we employ Schur duality of $\cH^{\otimes n}$. 
For this aim, we prepare several notations.
A sequence of monotone-increasing non-negative integers 
$\blambda=(\lambda_1,\ldots,\lambda_d)$ is called Young index.
Although many references \cite{H-q-text,Group1,GW} define Young index as 
monotone-decreasing non-negative integers,
we define it in the opposite way for notational convenience. 
We denote the set of Young indices $\blambda$ with the condition
$\sum_{j=1}^d\lambda_j=n$ by $Y_d^n$ .
Then, as explained in \cite[Section 6.2]{H-q-text}, we have
\begin{align}
\cH^{\otimes n}=
\bigoplus_{\blambda \in Y_d^n}
{\cal U}_{\blambda} \otimes {\cal V}_{\blambda}.
\end{align}
Here,
${\cal U}_{\blambda}$ expresses the irreducible space of $\SU(d)$
and ${\cal V}_{\blambda}$ expresses the irreducible space of  
the representation $\pi$ of the permutation group $\frS_n$.
We define 
$d_{\blambda}:= \dim {\cal U}_{\blambda}$,
$d_{\blambda}':= \dim {\cal V}_{\blambda}$.
Then, for states $\sigma$, $\rho$, we have
\begin{align}
\rho^{\otimes n}&=
\sum_{\blambda \in Y_d^n}
\rho_{\blambda} \otimes \rho_{\blambda,\mix}, \Label{BBR}\\
\sigma^{\otimes n}&=
\sum_{\blambda \in Y_d^n}
\sigma_{\blambda} \otimes \rho_{\blambda,\mix}, \Label{BBR2}
\end{align}
where $\rho_{\blambda,\mix}$ is the completely mixed state on ${\cal V}_{\blambda}$.
In the above equations, we consider that 
matrices on the subspace ${\cal U}_{\blambda} \otimes {\cal V}_{\blambda} $
are matrices on the whole space $\cH^{\otimes n}$.
In the first step, we apply the Schur transform $U_{Schur}$, which realizes
a basis change to block-diagonalize $U^{\otimes n}$ for an arbitrary unitary $U$.
It is known that the Schur transform $U_{Schur}$ can be
efficiently done \cite{BCH}.
We choose a basis 
$\{|u_{\blambda,j}\rangle\}_{j=1}^{d_{\blambda}}$ 
of ${\cal U}_{\blambda}$
such that $|u_{\blambda,j}\rangle$ is an eigenvector of $\sigma_{\blambda} $.

To estimate the quantum relative entropy $D(\rho\|\sigma)$, we employ the POVM $M=\{M_{\blambda,j}\}_{\blambda,j}$ defined as
\begin{align}
M_{\blambda,j}:= |u_{\blambda,j}\rangle \langle u_{\blambda,j}|
\otimes 
I( {\cal V}_{\blambda}),
\end{align}
where $I( {\cal V}_{\blambda})$ is the identity operator on ${\cal V}_{\blambda}$.
The measurement outcome is given as $(\blambda,j)$ so that 
it obeys the probability distribution
\begin{align}
P_{\rho,\sigma}^{(n)}(\blambda,j):=\Tr M_{\blambda,j} \rho^{\otimes n}.
\end{align}
To estimate the quantum relative entropy $D(\rho\|\sigma)$, 
we employ the estimate $x_{\blambda,j}$ defined as
$-\frac{1}{n}\log \Tr M_{\blambda,j} \sigma^{\otimes n}$.
Our estimation protocol is summarized as Protocol \ref{estimation-protocol}.

\begin{algorithm}[H]
\floatname{algorithm}{Protocol}
\caption{Estimation of relative entropy}
\label{estimation-protocol}
    \begin{enumerate}[label=\textbf{Step \arabic*.}]
        \setlength{\leftskip}{12mm}
        \item We apply the Schur transform $U_{Schur}$.
        \item We measure the system by using the measurement $M=\{M_{\blambda,j}\}_{\blambda,j}$.
        \item We set our estimate to be $x_{\blambda,j}:=-\frac{1}{n}\log \Tr M_{\blambda,j} \sigma^{\otimes n}$.
    \end{enumerate}
\end{algorithm}

To evaluate the mean square error of this estimator,
we denote the dimension of 
$\oplus_{\blambda \in Y_d^n}
{\cal U}_{\blambda}$ by $d_{n,d}$.
Since as shown in \cite[(6.16) and (6.18)]{H-q-text}, we have
\begin{align}
\dim {\cal U}_{\blambda} \le (n+1)^{d(d-1)/2},\quad
|Y_d^n|\le (n+1)^{d-1},
\Label{NM1}
\end{align}
the relation
\begin{align}
d_{n,d} \le  
(n+1)^{(d+2)(d-1)/2}\Label{BVT}
\end{align}
holds.

The reason of our choices of 
the POVM $M_{\blambda,j}$ and the estimate $x_{\blambda,j}$ is the following.
The reference \cite{H-01} essentially showed the relation
\begin{align}
& D(\rho^{\otimes n}\|\sigma^{\otimes n}) 
- \frac{d(d-1)}{2}\log (n+1) \notag\\
\le &
D_M(\rho^{\otimes n}\|\sigma^{\otimes n}) 
\le D(\rho^{\otimes n}\|\sigma^{\otimes n}),\Label{NBW}
\end{align}
where
\begin{align*}
&D_M(\rho^{\otimes n}\|\sigma^{\otimes n}) \\
:=& 
\sum_{\blambda,j}
\Tr M_{\blambda,j} \rho^{\otimes n}
(\log \Tr M_{\blambda,j} \rho^{\otimes n}
-\log \Tr M_{\blambda,j} \sigma^{\otimes n}).
\end{align*}
This relation shows that
the POVM $M_{\blambda,j}$ almost preserves the information with respect to the relative entropy
$D(\rho^{\otimes n}\|\sigma^{\otimes n})$.
Further, 
since the Shannon entropy 
$H(P_{\rho,\sigma}^{(n)}(\blambda,j)) 
:=-\sum_{\blambda,j}
\Tr M_{\blambda,j} \rho^{\otimes n}
\log \Tr M_{\blambda,j} \rho^{\otimes n} $ satisfies
\begin{align}
0\le H(P_{\rho,\sigma}^{(n)}(\blambda,j)) \le \log d_{n,d},
\end{align}
the the combination of \eqref{BVT} and \eqref{NBW}
implies
\begin{align}
& D(\rho^{\otimes n}\|\sigma^{\otimes n}) 
-(d+1)(d-1)\log (n+1)
\nonumber \\
\le&-\sum_{\blambda,j}
\Tr M_{\blambda,j} \rho^{\otimes n}
\log \Tr M_{\blambda,j} \sigma^{\otimes n} \notag\\
\le & D(\rho^{\otimes n}\|\sigma^{\otimes n}) .
\end{align}
That is, the average of $-\frac{1}{n}\log \Tr M_{\blambda,j} \sigma^{\otimes n} $
is close to the value $D(\rho\|\sigma)$ to be estimated.
This is the reason why we choose
the estimate $x_{\blambda,j}$.

The mean square error of this estimator is evaluated as follows.
\begin{theorem}\label{THB}
For any states $\rho$ and $\sigma$,
we have
\begin{align}
&\MSE_n(\rho\|\sigma):=\sum_{(\blambda,j)}
P_{\rho,\sigma}^{(n)}(\blambda,j)
(x_{\blambda,j}-D(\rho\|\sigma))^2 \notag\\
\le &
\Big(\frac{1}{\sqrt{n}}\sqrt{V(\rho\|\sigma)}+
\frac{1}{n}\log d_{n,d}\Big)^2.
\Label{BNB}
\end{align}
\end{theorem}

This theorem is proved in Appendix \ref{Sec4}.
The difference between the upper bound \eqref{BNB}
and the leading term $\frac{1}{n}V(\rho\|\sigma)$
is $\frac{1}{n^2}(\log d_{n,d})^2+
\frac{2}{n \sqrt{n}} \sqrt{V(\rho\|\sigma)}\log d_{n,d} $.
Using \eqref{BVT}, this different is upper bounded by
$\frac{1}{n^2} (\frac{(d+2)(d-1)}{2}
\log (n+1))^2+\frac{2}{n \sqrt{n}} 
\sqrt{V(\rho\|\sigma)}
\frac{(d+2)(d-1)}{2} \log (n+1)$.
$n$ times of these terms go to zero.

Next, we focus on the tail probability for our estimator, i.e.,
we study the estimation error probability as
\begin{align}
\delta_{\rho,\sigma}^{(n),+,\epsilon}
&:=P_{\rho,\sigma}^{(n)}
\{(\blambda,j): x_{\blambda,j}-D(\rho\|\sigma)>\epsilon
\} \\
\delta_{\rho,\sigma}^{(n),-,\epsilon}
&:=P_{\rho,\sigma}^{(n)}
\{(\blambda,j): x_{\blambda,j}-D(\rho\|\sigma)<-\epsilon
\} .
\end{align}
When $\epsilon$ behaves as 
$c/\sqrt{n}$, as shown in Appendix \ref{Sec5} 
the probabilities
$\delta_{\rho,\sigma}^{(n),\pm, c/\sqrt{n}}$
converges to a normal distribution as follows.

\begin{theorem}\label{TH3}
For any states $\rho$ and $\sigma$,
we have
\begin{align}
\delta_{\rho,\sigma}^{(n),+, c/\sqrt{n}}
&\to \int_{c/\sqrt{V(\rho\|\sigma)}}^\infty
\frac{1}{\sqrt{2\pi}}e^{-\frac{x^2}{2}}
dx \\
\delta_{\rho,\sigma}^{(n),-, c/\sqrt{n}}
&\to \int^{c/\sqrt{V(\rho\|\sigma)}}_{-\infty}
\frac{1}{\sqrt{2\pi}}
e^{-\frac{x^2}{2}} dx.
\end{align}
\end{theorem}
This property is called local asymptotic normality.

When the error term $\epsilon$ is constant, 
the tail probability goes to zero exponentially.
We use the sandwich relative entropy
$D_{\alpha}(\sigma\|\rho):=\frac{1}{\alpha-1}\log 
\Tr (\rho^{\frac{1-\alpha}{2\alpha}}
\sigma \rho^{\frac{1-\alpha}{2\alpha}})^\alpha$.
As shown in Appendix \ref{Sec6},
using the method in Section X of \cite{H-24}, we have the following theorem.

\begin{theorem}\label{TH1}
For any states $\rho$, $\sigma$, and $R,r>0$,
we have
\begin{align}
\delta_{\rho,\sigma}^{(n),+, R-D(\rho\|\sigma)}
\le & (d_{n,d})^{\alpha}e^{-n \alpha (D_{1-\alpha}(\rho\|\sigma)-R)},\Label{BN1}
\\
\delta_{\rho,\sigma}^{(n),-, R-D(\rho\|\sigma)}
\le & e^{-n (R-r-\alpha (D_{1+\alpha}(\rho\|\sigma) )}\notag\\
&+d_{n,d} e^{-nr} \Label{BN2}.
\end{align}
In \eqref{BN1}, $\alpha\in (0,1)$, and 
in \eqref{BN2}, $\alpha>0$.
\end{theorem}

Theorem \ref{TH1} guarantees that 
$\delta_{\rho,\sigma}^{(n),+,\epsilon}$ and
$\delta_{\rho,\sigma}^{(n),-,\epsilon}$ go to zero exponentially.
This property is called the large deviation.

\section{Analysis with increasing $d$}\Label{Sec3B}
Similar to \cite{AISW,WZ},
we consider the case when the dimension $d$ increases.
When $n=O(d^{4+\delta})$,
$\frac{1}{\sqrt{n}}\log d_{n,d}
\le 
\frac{1}{\sqrt{n}} \frac{(d+1)(d-1)}{2}\log (n+1)\to 0$.
In this case, due to Theorem \ref{THB},
the Cram\'{e}r-Rao type bound can be attained
because the second term of RHS in \eqref{BNB} goes to zero.

We now discuss the behavior of $n$ with respect to $d$ that allows the error to converge to zero as $d \to \infty$, rather than assuming $n = O(d^{4+\delta})$.
We denote the $d$-dimensional system by ${\cal H}_d$.
To discuss this case, we consider the sequence of states $\sigma_d$ depending on the dimension $d$.
We also choose a subset ${\cal S}_d$ of the set ${\cal S}({\cal H}_d)$ 
of states on ${\cal H}_d$.
We assume the following condition for $\sigma_d$ and ${\cal S}_d$;
\begin{align}
c_0:=\lim_{d \to \infty}\frac{1}{d^2}\max_{\rho \in {\cal S}_d}V(\rho \|\sigma_d)< \infty.\Label{CNB}
\end{align}

For example, 
when $\sigma_d$ is fixed to the completely mixed state $\rho_{\mix}$,
$V(\rho\|\rho_{\mix})$ equals 
the varentropy $\Tr \rho (\log \rho)^2-(\Tr \rho \log \rho)^2$.
Since Lemma 8 of \cite{H-02} showed that 
$\Tr \rho (\log \rho)^2 \le (\log d)^2$,
the condition \eqref{CNB} holds even with
${\cal S}_d ={\cal S}({\cal H}_d)$
so that the term $\frac{1}{n}\log d_{n,d}$
is the major term in the RHS of \eqref{BNB}.
Further,
the condition \eqref{CNB} holds with 
${\cal S}_d ={\cal S}({\cal H}_d)$
under certain conditions for $\sigma_d $.

\begin{lemma}\Label{LLO}
When the minimum eigenvalue of $\sigma_d $ is lower bounded by 
$e^{-td}$ with a fixed real number $t>0$,
the condition \eqref{CNB} holds
even when ${\cal S}_d $ equals ${\cal S}({\cal H}_d)$.
\end{lemma}

\begin{proof}
Considering the definition of $V(\rho\|\sigma_d)$ and the norm $\|~\|_\rho$,
we have
\begin{align}
&\sqrt{V(\rho\|\sigma_d)}\notag\\
=&\| \log \rho-(\Tr \rho \log \rho)
- \log \sigma_d +(\Tr \rho \log \sigma_d)\|_{\rho}  \notag\\
\le & 
\| \log \rho-(\Tr \rho \log \rho)\|_{\rho}  \notag\\
&+\| \log \sigma_d -(\Tr \rho \log \sigma_d)\|_{\rho}  .\Label{CX1}
\end{align}
Using the property of the norm $\|~\|_\rho$,
we have
\begin{align}
& \| \log \rho-(\Tr \rho \log \rho)\|_{\rho} ^2 \notag\\
=& \Tr \rho (\log \rho-(\Tr \rho \log \rho))^2 \notag\\
=& (\Tr \rho (\log \rho)^2)-(\Tr \rho \log \rho)^2 \notag\\
\le &  (\Tr \rho (\log \rho)^2)
\stackrel{(a)}{\le}  (\log d)^2, \Label{CX2}
\end{align}
where $(a)$ follows from \cite[Lemma 8]{H-02}.
Similarly, we have
\begin{align}
& \| \log \sigma_d -(\Tr \rho \log \sigma_d)\|_{\rho}^2 \notag\\
=& \Tr \rho (\log \sigma_d-(\Tr \rho \log \sigma_d))^2 \notag\\
=& (\Tr \rho (\log \sigma_d)^2)-(\Tr \rho \log \sigma_d)^2 \notag\\
\le &  (\Tr \rho (\log \sigma_d)^2)
\le (t d)^2\Label{CX3}.
\end{align}
Combining \eqref{CX1}, \eqref{CX2}, and \eqref{CX3}, we obtain the condition \eqref{CNB}.
\end{proof}

The sample complexity is characterized as the following theorem.
\begin{theorem}\label{TH4}
When the condition \eqref{CNB} holds and $n=c d^2$,
we have
\begin{align}
&\lim_{d\to \infty}\max_{\rho \in {\cal S}_d}
\delta_{\rho,\sigma}^{(n),+, \epsilon}
+\delta_{\rho,\sigma}^{(n),-, \epsilon}\notag\\
 \le& \frac{1}{\epsilon^2}
 \Big( \frac{\sqrt{c_0}}{\sqrt{c}}+\min_{0<s<1} \frac{c^{s-1}  }{s (1-s)}\Big)^2 
 \notag\\
 \le &\frac{(\sqrt{c_0}+4)^2  }{c \epsilon^2}.
\end{align}
\end{theorem}

When we choose $c$ to be $\frac{(\sqrt{c_0}+4)^2  }{\epsilon' \epsilon^2}$ for any $\epsilon,\epsilon'>0$, 
i.e., $n= \frac{(\sqrt{c_0}+4)^2  }{\epsilon' \epsilon^2} d^2$,
we have
\begin{align}
\lim_{d\to \infty}\max_{\rho \in {\cal S}_d}
\delta_{\rho,\sigma}^{(n),+, \epsilon}
+\delta_{\rho,\sigma}^{(n),-, \epsilon}
 \le \epsilon'.\Label{NMC}
\end{align}
Hence, when the condition \eqref{NMC} is imposed for the estimation error, 
we obtain an upper bound $O(d^2)$
of the sample complexity for the estimation of 
$D(\rho\|\sigma_d)$ when $\rho$ belongs to ${\cal S}_d$. 

Now, we compare the sample complexity of our method with the sample complexity of the simple application of full tomography \cite{HHJ,FO}.
In particular, we employ the error analysis on the trace norm \cite{HHJ} rather than other error criteria because the trace norm is most useful to evaluate the error of the relative entropy.
Their method achieves the estimation error
\begin{align}
\|\hat{\rho}-\rho \|_1 \le \epsilon
\end{align}
with high probability by sample complexity $O(d^2/\epsilon^2)$,
where $\hat{\rho}$ is the estimated density and $\rho$ is the true density.
When the known density $\sigma_d$ is the completely mixed state,
our problem is reduced to the estimation of von Neumann entropy
$S(\rho)$. Fannes inequality implies the relation
\begin{align}
| S(\hat{\rho})-S(\rho)|\le \|\hat{\rho}-\rho \|_1 \log d.
\end{align} 
That is, to achieve the condition
\begin{align}
| S(\hat{\rho})-S(\rho)|\le \epsilon',
\end{align} 
we need the sample complexity $O(d^2 (\log d)^2 /{\epsilon'}^2)$.
Although our method and the method by \cite{AISW} improves 
the sample complexity over the above method,
the degree of the improvement is incremental
in this case.

When the known state $\sigma_d$ is 
a general state to satisfy the condition of Lemma \ref{LLO}, i.e., 
the minimum eigenvalue of $\sigma_d $ is lower bounded by 
$e^{-td}$ with a fixed real number $t>0$,
the part $\log \sigma_d$ has the operator norm $t d$ at most.
Hence, we have
\begin{align}
& | D(\hat{\rho}\|\sigma_d)-D(\rho\|\sigma_d)| \notag\\
\le & | S(\hat{\rho})-S(\rho)|
+ | \Tr (\hat{\rho} -\rho)\log \sigma_d| \notag\\
\le & \|\hat{\rho}-\rho \|_1 \log d+ \|\hat{\rho}-\rho \|_1 \|\log \sigma_d\|
\notag\\
\le &\|\hat{\rho}-\rho \|_1(\log d+ t d ).
\end{align} 
That is, to achieve the condition
\begin{align}
| D(\hat{\rho}\|\sigma_d)-D(\rho\|\sigma_d)| \le \epsilon',
\end{align} 
we need the sample complexity $O(d^2 (\log d+ t d )^2 /{\epsilon'}^2)
=O(d^4 /{\epsilon'}^2)$.
Our method greatly improves the sample complexity over the 
the simple application of full tomography \cite{HHJ}.
That is, our extension to the estimation of the relative entropy
realizes significant improvement in the sample complexity over the above method.

In fact, the paper \cite{FO} also considers the sample
complexity when Bures $\chi^2$-divergence
and Bures distance are adopted.
However, the evaluation of the error under these criteria does not fit well with the above method.
Moreover,
it is challenging to find 
a good upper bound of the estimation errors 
of $S(\hat{\rho})-S(\rho)$
and $D(\hat{\rho}\|\sigma)-D(\rho\|\sigma)$
by using the estimation error based on these criteria.

\begin{proof}[Proof of Theorem \ref{TH4}]
We choose $s \in (0,1)$. Using
Weyl's dimension formula \cite[Theorem 7.1.7]{GW}, we have
\begin{align}
&\log \dim {\cal U}_{\blambda} 
=\log \prod_{1 \le i < j \le d}\frac{j-i+\lambda_j-\lambda_i}{j-i}\notag\\
=&\!\!\sum_{1 \le i < j \le d}\!\!\log \Big(1+\frac{\lambda_j-\lambda_i}{j-i}\Big)
\notag\\
=&\sum_{l=1}^d \sum_{i=1}^{d-l}
\log \Big(1+\frac{\lambda_{l+i}-\lambda_i}{l}\Big) \notag\\
=&\sum_{l=1}^d (d-l)
\sum_{i=1}^{d-l} \frac{1}{d-l}
\log \Big(1+\frac{\lambda_{l+i}-\lambda_i}{l} \Big) \notag\\
\stackrel{(a)}{\le} &\sum_{l=1}^d (d-l)
\log \Big(1+
\sum_{i=1}^{d-l} \frac{1}{d-l}
\frac{\lambda_{l+i}-\lambda_i}{l}\Big) \notag\\
\le &\sum_{l=1}^d (d-l)
\log \Big(1+
\frac{n}{(d-l)l}\Big) \notag\\
\stackrel{(b)}{\le} &
 \sum_{l=1}^d \frac{(d-l)}{s}
\Big(\frac{n}{(d-l)l}\Big)^s 
=  \sum_{l=1}^d \frac{(d-l)^{1-s} n^s }{s l^s},\Label{ZOI}
\end{align}
where $(a)$ follows from Jensen inequality for 
the concave function $x \mapsto \log (1+x)$ and 
the uniform distribution over $\{1, \ldots, d-l\}$.
$(b)$ can be shown as follows.
\begin{align}
& \log \Big(1+
\frac{n}{(d-l)l}\Big)
=\frac{1}{s}
\log \Big(1+
\frac{n}{(d-l)l}\Big)^s \notag\\
\stackrel{(c)}{\le} & 
\frac{1}{s} \log \Big(1+
\Big(\frac{n}{(d-l)l}\Big)^s\Big) 
\stackrel{(d)}{\le} 
\frac{1}{s}
\Big(\frac{n}{(d-l)l}\Big)^s, \Label{BNC}
\end{align}
where
$(c)$ follows from the inequality $(x+y)^s \le x^s+y^s$ for $x,y\ge 0$, and
$(d)$ follows from the inequality $\log (1+x)\le x$.
Thus, we have
\begin{align}
& \log d_{n,d} 
\le\log |Y_d^n| \max_{\blambda}\dim {\cal U}_{\blambda} \notag\\
\le& \log |Y_d^n| +\sum_{l=1}^d \frac{(d-l)^{1-s} n^s }{s l^s} \notag\\
\le& (d-1) \log (n+1) +\sum_{l=1}^d \frac{(d-l)^{1-s} n^s }{s l^s}.
\end{align}

When $n=c d^2$, we have
\begin{align}
& \frac{1}{d^2}\log \dim {\cal U}_{\blambda}  
\le 
\frac{1}{d^2}\sum_{l=1}^d \frac{(d-l)^{1-s} c^sd^{2s} }{s l^s} \notag\\
\to &
\int_0^1 \frac{(1-x)^{1-s} c^s  }{s x^s} dx\notag\\
\le  &
\int_0^1 \frac{c^s  }{s x^s} dx
=  
\Big[\frac{c^s  }{s (1-s)} x^{1-s} \Big]_0^1=\frac{c^s  }{s (1-s)}.
\Label{NM3}
\end{align}
Since the term 
$\frac{1}{n}(d-1) \log (n+1)$ goes to zero, 
the term $\frac{1}{n}\log d_{n,d}$ is evaluated as
\begin{align}
\lim_{d\to \infty}\frac{1}{n}\log d_{n,d} \le \min_{0<s<1}\frac{c^{s-1}  }{s (1-s)}\Label{CV1}.
%+o(1)
\end{align}
The definition of $c_0$ implies
\begin{align}
\lim_{d\to \infty}
\frac{1}{\sqrt{n}}\max_{\rho \in {\cal S}_d}\sqrt{V(\rho\|\sigma_d)}
=\frac{\sqrt{c_0}}{\sqrt{c}}.\Label{CV2}
\end{align}
Thus, we have
\begin{align}
& \lim_{d\to \infty}\max_{\rho \in {\cal S}_d}
\MSE_n(\rho\|\sigma_d) \notag \\
\stackrel{(e)}{\le} & 
 \Big( \frac{\sqrt{c_0}}{\sqrt{c}}+\min_{0<s<1} \frac{c^{s-1}  }{s (1-s)}\Big)^2 
\stackrel{(f)}{\le} \frac{(\sqrt{c_0}+4)^2  }{c},
\end{align}
where $(e)$ follows from the combination of 
\eqref{CV1}, \eqref{CV2}, and Theorem \ref{THB},
$(f)$ follows by choosing $s=1/2$.
Finally, applying Chebyshev's inequality,
we obtain the desired inequalities.
\end{proof}

The key point of this proof is the derivation \eqref{ZOI}.
In particular, the derivation \eqref{BNC} is the heart in this part.
The idea of this derivation comes from the proof of (12) given in Appendix IV
of the paper \cite{Ha11-3}, which gives the security analysis in information-theoretic security.

\section{Analysis on time complexity}\Label{Sec3C}
To implement our estimator, 
we need to apply suitable unitary gates before the measurement based on the computation basis.
For simplicity, we consider the case when 
$\sigma$ is a diagonal matrix.
For a general density $\sigma$, we need to apply a unitary 
diagonalizing  the density $\sigma$ to each system ${\cal H}$
before the following procedure.

The references \cite{BCH,Krovi,Nguyen,GBO} proposed  quantum algorithms for the Schur transform $U_{Schur}$.
We evaluate its precision as follows.
Given a required error $\epsilon>0$ and a  target unitary $V$,
we say that a unitary $U$ is an $\epsilon$-approximate of $V$ 
when any pure state $|\psi\rangle$ satisfies
\begin{align}
\| (U-V)|\psi\rangle \| \le \epsilon.\Label{BZW}
\end{align}
When the constructed unitary is 
an $\epsilon$-approximate of the Schur transform $U_{Schur}$, 
the time complexity of the construction \cite{Nguyen} is 
$O(n d^4 \polylog (n,d,1/\epsilon))$.
Since these constructions realize the approximation \eqref{BZW},
the estimation error probability increases at most only $\epsilon^2$.
We denote the distribution of $(\lambda,j)$ 
when we employ an $\epsilon$-approximate of the Schur transform $U_{Schur}$
by 
$P_{\rho,\sigma}^{(n),\epsilon}$.

In the classical part of our estimator, we need to calculate
$x_{\blambda,j}=-\frac{1}{n}\log \Tr M_{\blambda,j} \sigma^{\otimes n}$.
When $ \sigma$ has diagonal elements $s_1, \ldots,s_d $,
the eigenvalue of $\sigma^{\otimes n}$ on the projection
$M_{\blambda,j}$ has the form 
$\prod_{j=1}^{d} s_j^{n_j}$, where $\sum_{j=1}^d n_j=n$.
Then, we  have
\begin{align}
\Tr M_{\blambda,j} \sigma^{\otimes n}= 
\dim {\cal V}_{\blambda}
\prod_{j=1}^{d} s_j^{n_j}.
\end{align}

For a Young index  $\blambda$, the dimension ${\cal V}_{\blambda}$
is evaluated as 
\begin{align}
%e^{n H (\frac{\blambda}{n})}\le 
\dim {\cal V}_{\blambda}& =
%e^{n H (\frac{\blambda}{n})}
\frac{n!}{\blambda!}
e(\blambda)\\
e(\blambda)&:=\prod_{j>i}\frac{\lambda_j-\lambda_i-i+j}{\lambda_j+j-i}<1,
\end{align}
where $ \blambda!:=\lambda_1!\lambda_2!\cdots \lambda_d!$.
For its detail, see \cite[Eq. (2.72)]{Group1}.
Using \cite[Theorem 2.5]{Springer}, we have 
\begin{align}
e^{n H (\frac{\blambda}{n})} (n+1)^{-(d-1)}
\le \frac{n!}{\blambda!}\le 
e^{n H (\frac{\blambda}{n})}.
\end{align}
The value $e(\blambda)$ is evaluated as
\begin{align}
&-\log e(\blambda)
=\sum_{j>i}\log \frac{\lambda_j+j-i}{\lambda_j-\lambda_i-i+j} \notag\\
=&\sum_{j>i}\log (1 +\frac{\lambda_i}{\lambda_j-\lambda_i+j-i}) \notag\\
\le & \sum_{j>i}\log (1 +\frac{\lambda_i}{j-i})\notag\\
=&\sum_{l=1}^d (d-l)
\sum_{i=1}^{d-l} \frac{1}{d-l}
\log \Big(1+\frac{\lambda_i}{l} \Big) \notag\\
\stackrel{(a)}{\le} &\sum_{l=1}^d (d-l)
\log \Big(1+
\sum_{i=1}^{d-l} \frac{1}{d-l}
\frac{\lambda_i}{l}\Big) \notag\\
\le &\sum_{l=1}^d (d-l)
\log \Big(1+
\frac{n}{(d-l)l}\Big) \notag\\
\stackrel{(b)}{\le} &
\sum_{l=1}^d \frac{(d-l)^{1-s} n^s }{s l^s},
\end{align}
where $(a)$ follows from Jensen inequality for 
the concave function $x \mapsto \log (1+x)$, and 
$(b)$ follows from \eqref{BNC} in the same way as \eqref{ZOI}.

Instead of $x_{\blambda,j}$, we use its approximation
$x_{\blambda,j}^*:=
H (\frac{\blambda}{n})-\sum_{j=1}^d \frac{n_j}{n} \log s_j  $,
whose time complexity is $O(n)$.
Then, we have
\begin{align}
0 \le & x_{\blambda,j}^*-x_{\blambda,j}
\le 
 \frac{1}{n}(d\log (n+1)-\log e(\blambda) )\notag \\
\le &
 \frac{1}{n} \Big(d\log (n+1)-
 \sum_{l=1}^d \frac{(d-l)^{1-s} n^s }{s l^s}\Big) 
\Label{NM2}.
\end{align}

When $n=cd^2 $,
$\lim_{n\to \infty}
 \frac{1}{n}(d\log (n+1)-
 \sum_{l=1}^d \frac{(d-l)^{1-s} n^s }{s l^s}) 
 \le \frac{4}{\sqrt{c}} $.

Given an arbitrary number $\epsilon>0$, we choose $c=\frac{16}{\epsilon^2}$.
When $n=cd^2$,
the limit of the upper bound in \eqref{NM2} can be bounded by $\epsilon$.
Then, we define 
\begin{align}
\delta_{\rho,\sigma}^{(n),+,\epsilon,\epsilon_1}
&:=P_{\rho,\sigma}^{(n),\epsilon_1}
\{(\blambda,j): x_{\blambda,j}^*-D(\rho\|\sigma)>\epsilon
\} \\
\delta_{\rho,\sigma}^{(n),-,\epsilon,\epsilon_1}
&:=P_{\rho,\sigma}^{(n),\epsilon_1}
\{(\blambda,j): x_{\blambda,j}^*\!-\!D(\rho\|\sigma)
\!<\!-\epsilon
\} .
\end{align}
Thus, when $n=c d^2$, using \eqref{NMC}, 
we have
\begin{align}
&\lim_{d\to \infty}\max_{\rho \in {\cal S}_d}
\delta_{\rho,\sigma}^{(n),+, \epsilon,\epsilon_1}
+\delta_{\rho,\sigma}^{(n),-, \epsilon,\epsilon_1}\notag\\
 \le & \frac{(\sqrt{c_0}+4)^2  }{c(\epsilon-\frac{4}{\sqrt{c}} )^2}
 +\epsilon_1^2
  \Label{3NMC}.
\end{align}

Hence, when we choose $\epsilon_1=\sqrt{\epsilon'/2}$
and $c=\min(\frac{64}{\epsilon^2},\frac{8(4+\sqrt{c_0})^2}{\epsilon^2\epsilon'})$ for any $\epsilon,\epsilon'>0$,
we have
\begin{align}
\lim_{d\to \infty}\max_{\rho \in {\cal S}_d}
\delta_{\rho,\sigma}^{(n),+, \epsilon,\epsilon_1}
+\delta_{\rho,\sigma}^{(n),-, \epsilon,\epsilon_1}
 \le \epsilon'.\Label{NM9}
\end{align}
Overall,
when $n=O(d^2)$, we use 
the construction of Schur transform $U_{Schur}$ by 
\cite{Nguyen} with the time complexity $O(d^6 \polylog (d,1/\epsilon))$
and the estimate $x_{\blambda,j}$ with 
the time complexity $O(n)=O(d^2)$.
That is, the total construction has a total time complexity
$O(d^6 \polylog (d))$ to realize the condition \eqref{NM9} with a fixed value $\epsilon>0$.

\begin{remark}
Initially, the paper \cite{BCH} gives a polynomial algorithm for Schur transform $U_{Schur}$. However, it did not give a precise evaluation of its time complexity. Although the paper \cite{Nguyen} writes that the protocol by the paper \cite{BCH} has time complexity $O\left(n d^3\polylog (n, d, 1 / \epsilon)\right)$ in its Table 1, the paper \cite{Nguyen} does not give the derivation of this evaluation because its 
main focus is 
the implementation of the mixed Schur transform 
$U_{MixedSchur}$, which realizes
a basis change to block-diagonalize $U^{\otimes n}\otimes \overline{U}^{\otimes m}$ for an arbitrary unitary $U$.
In fact, the paper \cite{Nguyen} proved that
its time complexity is 
$O\left((n+m) d^4 \polylog (n, m, d, 1 / \epsilon)\right)$, and
the recent paper \cite{GBO} also proposed the same algorithm as the algorithm by \cite{Nguyen}. 
In this paper, to make a safer statement, we employ the statement that the time complexity of the unitary $U_{MixedSchur}$ is $O\left((n+m) d^4 \polylog
(n, m, d, 1 / \epsilon)\right)$. Its special case with $m=0$ corresponds to $U_{Schur}$ and covers our case with the time complexity $O\left(n d^4 \polylog(n, d, 1 / \epsilon)\right)$.
\end{remark}

\begin{remark}
The paper \cite{WZ} also mentions that
there exists an algorithm for Schur transform $U_{Schur}$
whose time complexity is $O(d^6 \polylog (d,1/\epsilon))$
when $n=O(d^2)$ \cite{QW}.
First, they employ
the algorithm for weak Schur sampling given in \cite[Section 4.2.2]{MW}. 
Although this method employs the quantum Fourier transform 
the permutation group $\frS_n$,
they proposed to use the algorithm by \cite{KS} for this part.
The obtained algorithm has time complexity 
$O(n^3 \polylog (n,d,1/\epsilon))$.
When $n=d^2$, it equals $O(d^6 \polylog (d,1/\epsilon))$.
This implementation of weak Schur sampling was mentioned earlier in \cite[Section 1.1.2]{WZ2}.
\end{remark}

\begin{remark}
Recently, the paper \cite{SB}
studied the estimation of relative entropy with two unknown 
densities $\rho$ and $\sigma$
when an estimator of the state is given
and the estimator converges to the true density.
The analysis \cite{SB} is different from ours in the following points.
(i) We assume that $\sigma$ is known,
but the paper \cite{SB} assumes that it is unknown.
(ii) We discuss the optimality of our estimator
in various settings. 
Under the case with fixed $d$, we optimize
the MSE.
When $d$ increases, we discuss the sample complexity.
But,
when an estimator of the state is given,
the paper \cite{SB} 
discusses the behavior of the estimator of relative entropy 
by using the behavior of the estimator.
In addition, 
the paper \cite{SB} mentioned that the limiting distribution of 
the normalized estimator is the normal distribution.
When $d$ is fixed, the discussion explains how this problem is reduced to a special case of the parameter estimation with nuisance parameters.
The paper \cite{YCH} also showed 
that the limiting distribution of 
the normalized estimator is the normal distribution
in the parameter estimation with nuisance parameters.
\end{remark}

\section{Discussion}
We have studied the estimation of quantum relative entropy 
$D(\rho\|\sigma)$
when the second argument $\sigma$ is known.
Our estimator has a good performance in the following two senses.
First, when the dimension $d$ is fixed, our estimator attains Cram\'{e}r-Rao type bound asymptotically.
Second, when the dimension $d$ increases,
our estimator has sample complexity $O(d^2)$ under natural assumptions,
which greatly improves the sample complexity by the simple application of full tomography \cite{HHJ}.
Third, when $d$ is large,
although the calculation of our estimator is difficult,
we have an alternative estimator whose  
time complexity is $O(d^6 \polylog (d))$.
The difference between this alternative estimator and our original estimator approaches zero
asymptotically.
In the above way, our estimator works well unifiedly, 
which is a unique feature and did not appear in existing studies \cite{AISW,WZ,SYH}.

Since our estimator assumes the knowledge of the form of the second argument $\sigma$, it is needed to develop an estimator that works without this knowledge.
 In this case, we need to prepare $n$ copies of the unknown state $\sigma$ as well.
To attain the Cram\'{e}r-Rao type bound with fixed $d$, 
it is sufficient to estimate the form of the density matrix $\sigma$.
However, such a method has a large time complexity when $d$ increases.
It is an interesting open problem to derive an efficient estimator of the quantum relative entropy in the above setting. 

\section*{Acknowledgement}
The author is supported in part by the National Natural Science Foundation of China (Grant No.
62171212)
and 
the General R\&D Projects of 
1+1+1 CUHK-CUHK(SZ)-GDST Joint Collaboration Fund 
(Grant No. GRDP2025-022).
The author is grateful for Professor Kun Fang
to a helpful discussion.
Also, the author would like to thank Mr Dmitry Grinko and Dr. Qisheng Wang
for helpful discussions on the time complexity of 
Schur transform.
In addition, the author would like to thank 
Prof. Jun Suzuki and Dr. Francesco Albarelli
for helpful discussions on the Cram\'{e}r-Rao type bound.
Also, the author would like to thank 
Prof. Macro Tomamichel and Dr. Hu Yanglin 
for helpful discussions for the relation with full tomography
and informing the references \cite{HHJ,FO}.

\appendix
\section{Evaluation of mean square error: Proof of Theorem \ref{THB}}\Label{Sec4}
We define the state 
$\rho_{n}:= \sum_{\blambda} \rho_{\blambda}$ on the system $\oplus_{\blambda \in Y_d^n}
{\cal U}_{\blambda}$.
We choose an orthogonal basis $\{|v_{\blambda,j'}\}_{j'}$ of ${\cal V}_{\blambda}$.

We use the methods 
\cite{H-02,TH,HT}.
Then, we have the following calculation.
\begin{align}
&x_{\blambda,j}-D(\rho\|\sigma)
=-\frac{1}{n}\log \Tr M_{\blambda,j} \sigma^{\otimes n}
-D(\rho\|\sigma)\notag\\
=&
-\frac{1}{n}
(
\log \langle u_{\blambda,j},v_{\blambda,j'}| \sigma^{\otimes n}|u_{\blambda,j},v_{\blambda,j'}\rangle
+\log d_{\blambda}')\notag\\
&-D(\rho\|\sigma)\notag 
\end{align}
\begin{align}
=&
-\frac{1}{n}
(
\log \langle u_{\blambda,j},v_{\blambda,j'}| \sigma^{\otimes n}|u_{\blambda,j},v_{\blambda,j'}\rangle
\notag\\
&-
\log \langle u_{\blambda,j},v_{\blambda,j'}| \rho^{\otimes n}|u_{\blambda,j},v_{\blambda,j'}\rangle)
-D(\rho\|\sigma) \notag\\
&
-\frac{1}{n}
(\log d_{\blambda}'
+\log \langle u_{\blambda,j},v_{\blambda,j'}| \rho^{\otimes n}|u_{\blambda,j},v_{\blambda,j'}\rangle) \notag\\
=&
\frac{1}{n}
(
-\log \langle u_{\blambda,j},v_{\blambda,j'}| \sigma^{\otimes n}|u_{\blambda,j},v_{\blambda,j'}\rangle\notag\\
&+
\log \langle u_{\blambda,j},v_{\blambda,j'}| \rho^{\otimes n}
|u_{\blambda,j},v_{\blambda,j'}\rangle)
-D(\rho\|\sigma) \notag\\
&
-\frac{1}{n}
\log \langle u_{\blambda,j}| \rho_{\blambda}|u_{\blambda,j}\rangle.
\end{align}
Therefore, we have
\begin{align}
&\sum_{(\blambda,j)}
P_{\rho,\sigma}^{(n)}(\blambda,j)
(x_{\blambda,j}-D(\rho\|\sigma))^2 \notag\\
=&\sum_{(\blambda,j)}
P_{\rho,\sigma}^{(n)}(\blambda,j)
\Big(\frac{1}{n}
(
-\log \langle u_{\blambda,j}| \sigma_{\blambda}|u_{\blambda,j}\rangle\notag\\
&-D(\rho\|\sigma)
\Big)^2 \notag\\
=&\sum_{(\blambda,j,j')}
\langle u_{\blambda,j}, v_{\blambda,j'}| \rho^{\otimes n}
|u_{\blambda,j}, v_{\blambda,j'}\rangle
\notag\\
&\Big(\frac{1}{n}
(-\log 
\langle u_{\blambda,j}, v_{\blambda,j'}| \sigma_{\blambda}\otimes \rho_{\blambda,\mix} 
|u_{\blambda,j},  v_{\blambda,j'}\rangle
\notag \\
&+
\log \langle  v_{\blambda,j'}|  \rho_{\blambda,\mix} | v_{\blambda,j'}\rangle)
-D(\rho\|\sigma) \Big)^2 \notag\\
=&\sum_{(\blambda,j,j')}
\langle u_{\blambda,j}, v_{\blambda,j'}| \rho^{\otimes n}
|u_{\blambda,j}, v_{\blambda,j'}\rangle
\notag\\
&\Big(\frac{1}{n}
(-\log 
\langle u_{\blambda,j}, v_{\blambda,j'}| \sigma^{\otimes n} 
|u_{\blambda,j},  v_{\blambda,j'}\rangle \notag\\
&+\log \langle u_{\blambda,j}, v_{\blambda,j'}| 
\sum_{\blambda'}I_{{\cal U}_{\blambda}} \otimes
 \rho_{\blambda',\mix} | u_{\blambda,j},v_{\blambda,j'}\rangle
)
\notag \\
&
-D(\rho\|\sigma)
\Big)^2 \notag\\
=&
\Tr \rho^{\otimes n}
\Big(\frac{1}{n}
(-\log \sigma^{\otimes n}+
\sum_{\blambda} \log (I_{{\cal U}_{\blambda}} \otimes \rho_{\blambda,\mix}))
\notag \\
&
-D(\rho\|\sigma)
\Big)^2 \notag\\
=&
\Tr \rho^{\otimes n}
\Big(\frac{1}{n}
(-\log \sigma^{\otimes n}+
\sum_{\blambda} \log (\rho_{\blambda} \otimes \rho_{\blambda,\mix}))
\notag \\
&
-D(\rho\|\sigma)
-\sum_{\blambda} \log (\rho_{\blambda} \otimes I_{{\cal V}_{\blambda}})
\Big)^2 \notag\\
=&
\Tr \rho^{\otimes n}
\Big(\frac{1}{n}
(-\log \sigma^{\otimes n}+\log \rho^{\otimes n})
-D(\rho\|\sigma)
\notag \\
&
-\frac{1}{n}\sum_{\blambda}
\log (\rho_{\blambda} \otimes I_{{\cal V}_{\blambda}})
\Big)^2 \notag
\end{align}
\begin{align}
=&
\Big\| \frac{1}{n}
(-\log \sigma^{\otimes n}+\log \rho^{\otimes n})
-D(\rho\|\sigma) \notag\\
&-\frac{1}{n}\sum_{\blambda}
\log (\rho_{\blambda} \otimes I_{{\cal V}_{\blambda}})
\Big \|_{\rho^{\otimes n}}^2 \notag\\
\le&
\Big(
\Big\| \frac{1}{n}
(-\log \sigma^{\otimes n}+\log \rho^{\otimes n})
-D(\rho\|\sigma) \Big\|_{\rho^{\otimes n}}
\notag\\
&+\Big\| \frac{1}{n}\sum_{\blambda}
\log (\rho_{\blambda} \otimes I_{{\cal V}_{\blambda}})
\Big\|_{\rho^{\otimes n}}
\Big)^2 .\Label{BN1T}
\end{align}
Also, we have
\begin{align}
& \Big\| \frac{1}{n}
(-\log \sigma^{\otimes n}+\log \rho^{\otimes n})
-D(\rho\|\sigma) 
\Big \|_{\rho^{\otimes n}}\notag\\
=&\Tr \rho^{\otimes n}
\Big(\frac{1}{n}
(-\log \sigma^{\otimes n}+\log \rho^{\otimes n})
-D(\rho\|\sigma)\Big)^2
\notag \\
=&\frac{1}{n}V(\rho\|\sigma)  \Label{BN2T},
\end{align}
and 
\begin{align}
&\Big\|-\frac{1}{n}\sum_{\blambda}
\log (\rho_{\blambda} \otimes I_{{\cal V}_{\blambda}})
\Big
\|_{\rho^{\otimes n}}^2
\notag\\
=&\Tr \rho^{\otimes n}
\Big(\frac{1}{n}
\sum_{\blambda} \log (\rho_{\blambda} \otimes I_{{\cal V}_{\blambda}})
\Big)^2 \notag\\
=&\frac{1}{n^2} \Tr \rho_n (\log \rho_n)^2 
\stackrel{(a)}{\le } 
\frac{1}{n^2} (\log d_{n,d} )^2,\Label{BN3T}
\end{align}
where $(a)$ follows from \cite[Lemma 8]{H-02}.
The combination of \eqref{BN1T}, \eqref{BN2T}, and \eqref{BN3T} implies 
\eqref{BNB}.

\section{Local asymptotic normality: Proof of Theorem \ref{TH3}}\Label{Sec5}
The random variable $
\sqrt{n}(-\frac{1}{n}
(
\log \langle u_{\blambda,j},v_{\blambda,j'}| \sigma^{\otimes n}|u_{\blambda,j},v_{\blambda,j'}\rangle+
\log \langle u_{\blambda,j},v_{\blambda,j'}| \rho^{\otimes n}|u_{\blambda,j},v_{\blambda,j'}\rangle)
-D(\rho\|\sigma) )$ 
has additional variable $j'$.
The joint distribution of $(\blambda,j,j')$ is given as
$\langle u_{\blambda,j}, v_{\blambda,j'}| \rho^{\otimes n}
|u_{\blambda,j}, v_{\blambda,j'}\rangle$.

We define the pinching map $\Gamma_{\sigma,n}$ as
\begin{align}
&\Gamma_{\sigma,n}(\rho')\nonumber \\
:=&\sum_{\blambda\in Y_d^n, j}
(|u_{\blambda,j}\rangle \langle u_{\blambda,j}| \otimes P( {\cal V}_{\blambda}))
\rho'\notag\\
&\cdot(|u_{\blambda,j}\rangle \langle u_{\blambda,j}| \otimes P( {\cal V}_{\blambda}))
\end{align}
for a general state $\rho'$ on $\cH^{\otimes n}$,
where $P( {\cal V}_{\blambda})$ is the projection to ${\cal V}_{\blambda}$.
Then,
this random variable coincides with $
\sqrt{n}(\frac{1}{n}
(\log \Gamma_{\sigma,n}(\rho^{\otimes n})
- \log \sigma^{\otimes n})-D(\rho\|\sigma))
$ under the state 
$\Gamma_{\sigma,n}(\rho^{\otimes n})$
studied in Section V of \cite{TH}. 
The discussion in Section V of \cite{TH} showed that 
it converges to Gaussian distribution with variance 
$V(\rho\|\sigma)$ in probability as $n \to \infty$.
%Using \cite[(34)]{TH} and \cite[Theorem 19.]{HT}, 

In addition,
in the same way as \eqref{BN3T}, using \cite[Lemma 8]{H-02}, we have
\begin{align}
&\sum_{(\blambda,j)}
P_{\rho,\sigma}^{(n)}(\blambda,j)
%\frac{1}{\sqrt{n}}
(\frac{1}{\sqrt{n}}\log \langle u_{\blambda,j}| \rho_{\blambda}|u_{\blambda,j}\rangle)^2 \notag\\
=&\sum_{(\blambda,j)}
P_{\rho,\sigma}^{(n)}(\blambda,j)
%\frac{1}{\sqrt{n}}
(\frac{1}{\sqrt{n}}\log P_{\rho,\sigma}^{(n)}(\blambda,j))^2 \notag\\
=&\frac{1}{n}
\Tr \rho_n
(\log \rho_n)^2 
\le \frac{1}{n}(\log d_{n,d})^2 \notag\\
\le &
\frac{1}{n}(\frac{(d+2)(d-1)}{2}
\log (n+1))^2 \to 0.
\end{align}
Thus, the random variable 
$\frac{1}{\sqrt{n}}
\log \langle u_{\blambda,j}| \rho_{\blambda}|u_{\blambda,j}\rangle$
 under the distribution $P_{\rho,\sigma}^{(n)}$
 converges to zero in probability as $n \to \infty$.
Therefore, the probabilities 
$\delta_{\rho,\sigma}^{(n),\pm, c/\sqrt{n}}$
can be asymptotically
characterized by 
Gaussian distribution with variance $V(\rho\|\sigma)$.

\section{Large deviation evaluation}\Label{Sec6}
\subsection{Preparation}

Since the system ${\cal V}_{\blambda}$ has no information, it is sufficient to handle the two states
$\oplus_{\blambda \in Y_d^n}
\rho_{\blambda}$ and $
\oplus_{\blambda \in Y_d^n}
\sigma_{\blambda} 
$ over $\oplus_{\blambda \in Y_d^n}
{\cal U}_{\blambda}$.
We define the map
\begin{align}
\kappa_{\sigma,n}(\rho') %\nonumber \\
:=&\sum_{\blambda\in Y_d^n, j}
\Tr_{{\cal V}_{\blambda}}
(|u_{\blambda,j}\rangle \langle u_{\blambda,j}| \otimes P( {\cal V}_{\blambda}))
\rho'\notag \\
&\cdot(|u_{\blambda,j}\rangle \langle u_{\blambda,j}| \otimes P( {\cal V}_{\blambda}))
\end{align}
for a general state $\rho'$ on $\cH^{\otimes n}$.
Due to the above-mentioned structure,
the states $\Gamma_{\sigma,n}(\rho^{\otimes n})$ and $\sigma^{\otimes n}$
are commutative with each other.

We have the pinching inequality:
\begin{align}
\rho'\le d_{n,d} \Gamma_{\sigma,n}(\rho').\Label{NTD}
\end{align}
In fact, applying the matrix monotone function to \eqref{NTD},
for $\alpha \in (0,1)$, we have
\begin{align}
&d_{n,d}^{\alpha-1} 
((\sigma^{\otimes n})^{\frac{1-\alpha}{2\alpha}}
\Gamma_{\sigma,n}(\rho^{\otimes n})
(\sigma^{\otimes n})^{\frac{1-\alpha}{2\alpha}})^{\alpha-1}
\nonumber \\
=&
\Big(d_{n,d}
\Gamma_{\sigma,n}\Big(
(\sigma^{\otimes n})^{\frac{1-\alpha}{2\alpha}}
\rho^{\otimes n}
(\sigma^{\otimes n})^{\frac{1-\alpha}{2\alpha}} \Big)
\Big)^{\alpha-1} \nonumber \\
\le &
\Big(
(\sigma^{\otimes n})^{\frac{1-\alpha}{2\alpha}}
\rho^{\otimes n}
(\sigma^{\otimes n})^{\frac{1-\alpha}{2\alpha}} 
\Big)^{\alpha-1}.
\Label{NTD5}
\end{align}
Thus, we have
\begin{align}
&\Tr \big(\Gamma_{\sigma,n}(\rho^{\otimes n})\big)^{\alpha} 
(\sigma^{\otimes n})^{1-\alpha}
\nonumber \\
=&
\Tr 
\Big((\sigma^{\otimes n})^{\frac{1-\alpha}{2\alpha}}
\Gamma_{\sigma,n}(\rho^{\otimes n})
(\sigma^{\otimes n})^{\frac{1-\alpha}{2\alpha}} \Big)^\alpha \nonumber \\
=&
\Tr 
\Big(\Gamma_{\sigma,n}(
(\sigma^{\otimes n})^{\frac{1-\alpha}{2\alpha}}
\rho^{\otimes n}
(\sigma^{\otimes n})^{\frac{1-\alpha}{2\alpha}})
\Big)^\alpha \nonumber \\
=&
\Tr
\Big(
\Gamma_{\sigma,n}(
(\sigma^{\otimes n})^{\frac{1-\alpha}{2\alpha}}
\rho^{\otimes n}
(\sigma^{\otimes n})^{\frac{1-\alpha}{2\alpha}}) \Big)
\nonumber \\
&\cdot
 \Big(
\Gamma_{\sigma,n}(
(\sigma^{\otimes n})^{\frac{1-\alpha}{2\alpha}}
\rho^{\otimes n}
(\sigma^{\otimes n})^{\frac{1-\alpha}{2\alpha}})
\Big)^{\alpha-1} \nonumber \\
=&
\Tr 
\Big((\sigma^{\otimes n})^{\frac{1-\alpha}{2\alpha}}
\rho^{\otimes n}
(\sigma^{\otimes n})^{\frac{1-\alpha}{2\alpha}}\Big)
\nonumber \\
&\cdot \Big(\Gamma_{\sigma,n}(
(\sigma^{\otimes n})^{\frac{1-\alpha}{2\alpha}}
\rho^{\otimes n}
(\sigma^{\otimes n})^{\frac{1-\alpha}{2\alpha}})
\Big)^{\alpha-1} \nonumber \\
\stackrel{(a)}{\le} &
d_{n,d}^{1-\alpha}
\Tr 
\Big((\sigma^{\otimes n})^{\frac{1-\alpha}{2\alpha}}
\rho^{\otimes n}
(\sigma^{\otimes n})^{\frac{1-\alpha}{2\alpha}}\Big)
\nonumber \\
& \cdot \Big(
(\sigma^{\otimes n})^{\frac{1-\alpha}{2\alpha}}
\rho^{\otimes n}
(\sigma^{\otimes n})^{\frac{1-\alpha}{2\alpha}}
\Big)^{\alpha-1}\nonumber \\
=&
d_{n,d}^{1-\alpha}
\Tr 
\Big((\sigma^{\otimes n})^{\frac{1-\alpha}{2\alpha}}
\rho^{\otimes n}
(\sigma^{\otimes n})^{\frac{1-\alpha}{2\alpha}}\Big)^{\alpha}\nonumber \\
=&
d_{n,d}^{1-\alpha}
(\Tr 
(\sigma^{\frac{1-\alpha}{2\alpha}}
\rho
\sigma^{\frac{1-\alpha}{2\alpha}})^{\alpha})^n,\Label{NMR2}
\end{align}
where Step $(a)$ follows from \eqref{NTD5}.

\subsection{Proof of Theorem \ref{TH1}}
Here, we define the projection
$\{X >Y\}$ as the projection to 
the subspace spanned by the eigenvectors 
of $X-Y$ with positive eigenvalues.
For $\alpha \in (0,1)$, we have
\begin{align}
&\sum_{(\blambda,j): x_{\blambda,j} < R}
\Tr M_{\blambda,j} \rho^{\otimes n} \notag\\
=&
\Tr 
\kappa_{\sigma,n}(\rho^{\otimes n})
\{ \kappa_{\sigma,n}(\sigma^{\otimes n}) > e^{-nR}
\} \notag\\
\stackrel{(b)}{\le} &
\Tr 
\kappa_{\sigma,n}(\rho^{\otimes n})
\{ \kappa_{\sigma,n}(\sigma^{\otimes n}) > e^{-nR}
\kappa_{\sigma,n}(\rho^{\otimes n})\} \notag\\
\stackrel{(c)}{\le} &
\Tr 
\kappa_{\sigma,n}(\rho^{\otimes n})^{1-\alpha}
\kappa_{\sigma,n}(\sigma^{\otimes n})^{\alpha} e^{n \alpha R}
\notag\\
= &
\Tr 
\Gamma_{\sigma,n}(\rho^{\otimes n})^{1-\alpha}
(\sigma^{\otimes n})^{\alpha} e^{n \alpha R}
\notag\\
\stackrel{(d)}{\le} &
d_{n,d}^{\alpha} (\Tr (\sigma^{\frac{\alpha}{2(1-\alpha)}}
\rho \sigma^{\frac{\alpha}{2(1-\alpha)}})^{1-\alpha} e^{ \alpha R})^n\notag\\
=& d_{n,d}^{\alpha}e^{-n \alpha (D_{1-\alpha}(\rho\|\sigma)-R)}.
\end{align}
Here, $(b)$ follows from 
$\kappa_{\sigma,n}(\rho^{\otimes n})\le I$, and
$(d)$ follows from \eqref{NMR2}.
The relation $(c)$ follows from
\begin{align}
&\Tr 
\kappa_{\sigma,n}(\rho^{\otimes n})
\{ \kappa_{\sigma,n}(\sigma^{\otimes n}) > e^{-nR}
\kappa_{\sigma,n}(\rho^{\otimes n})\} \notag\\
=&
\Tr 
\kappa_{\sigma,n}(\rho^{\otimes n})
\{ 
e^{\alpha nR}
(\kappa_{\sigma,n}(\sigma^{\otimes n}))^\alpha  
> 
(\kappa_{\sigma,n}(\rho^{\otimes n}))^\alpha 
\} \notag\\
\stackrel{(e)}{\le} &
\Tr 
\kappa_{\sigma,n}(\rho^{\otimes n})^{1-\alpha}
(\kappa_{\sigma,n}(\sigma^{\otimes n}))^{\alpha} 
e^{n \alpha R} \notag\\
&\cdot \{ 
e^{\alpha nR}
(\kappa_{\sigma,n}(\sigma^{\otimes n}))^\alpha  
> 
(\kappa_{\sigma,n}(\rho^{\otimes n}))^\alpha 
\} \notag\\
{\le} &
\Tr 
\kappa_{\sigma,n}(\rho^{\otimes n})^{1-\alpha}
(\kappa_{\sigma,n}(\sigma^{\otimes n}))^{\alpha} e^{n \alpha R},
\end{align}
where $(e)$ follows from the following relation.
\begin{align}
\Tr &\Big[
(\kappa_{\sigma,n}(\rho^{\otimes n}))^{1-\alpha}
\notag\\
&\cdot\Big(
(\kappa_{\sigma,n}(\rho^{\otimes n}))^\alpha
-
(\kappa_{\sigma,n}(\sigma^{\otimes n}))^{\alpha} 
e^{n \alpha R} \Big)\notag\\
&\cdot \{ 
e^{\alpha nR}
(\kappa_{\sigma,n}(\sigma^{\otimes n}))^\alpha  
> 
(\kappa_{\sigma,n}(\rho^{\otimes n}))^\alpha 
\}\Big] \le 0.
\end{align}

For $\alpha >0$, we have
\begin{align}
&\sum_{(\blambda,j): x_{\blambda,j} > R}
\Tr M_{\blambda,j} \rho^{\otimes n} \notag\\
=&
\Tr 
\kappa_{\sigma,n}(\rho^{\otimes n})
\{ \kappa_{\sigma,n}(\sigma^{\otimes n}) < e^{-nR}
\} \notag\\
\le &
\Tr 
\kappa_{\sigma,n}(\rho^{\otimes n})
\{ \kappa_{\sigma,n}(\sigma^{\otimes n}) < e^{-nR}\} 
\{\kappa_{\sigma,n}(\rho^{\otimes n}) > e^{-nr} \} \notag\\
&+
\Tr 
\kappa_{\sigma,n}(\rho^{\otimes n})
\{\kappa_{\sigma,n}(\rho^{\otimes n}) \le e^{-nr} \} 
\notag\\
\le &
\Tr 
\kappa_{\sigma,n}(\rho^{\otimes n})
\{ \kappa_{\sigma,n}(\sigma^{\otimes n}) < e^{-nR+nr}
\kappa_{\sigma,n}(\rho^{\otimes n})\} \notag\\
%\{\kappa_{\sigma,n}(\rho^{\otimes n}) > e^{-nr} \} \\
&+
\Tr 
\kappa_{\sigma,n}(\rho^{\otimes n})
\{\kappa_{\sigma,n}(\rho^{\otimes n}) \le e^{-nr} \} 
\notag\\
\le &
\Tr 
\kappa_{\sigma,n}(\rho^{\otimes n})^{1+\alpha}
\kappa_{\sigma,n}(\sigma^{\otimes n})^{-\alpha} e^{-n \alpha (R-r)}
+d_{n,d} e^{-nr} 
\notag\\
= &
\Tr 
\Gamma_{\sigma,n}(\rho^{\otimes n})^{1+\alpha}
(\sigma^{\otimes n})^{-\alpha} e^{-n \alpha (R-r)}
+d_{n,d} e^{-nr} 
\notag\\
\stackrel{(f)}{\le} &
\Tr ( (\sigma^{\otimes n})^{-\frac{\alpha}{2(1+\alpha)}
\rho^{\otimes n} (\sigma^{\otimes n})^{-\frac{\alpha}{2(1+\alpha)}})^{1+\alpha} e^{- \alpha (R-r) n}}
\notag\\
%\Tr 
%(\rho^{\otimes n})^{1+\alpha}
%(\sigma^{\otimes n})^{-\alpha} e^{-n \alpha (R-r)}
&+d_{n,d} e^{-nr} 
\notag\\
= &
(\Tr (\sigma^{-\frac{\alpha}{2(1+\alpha)}}
\rho \sigma^{-\frac{\alpha}{2(1+\alpha)}})^{1+\alpha} e^{- \alpha (R-r)})^n
+d_{n,d} e^{-nr} \notag\\
=& e^{-n (R-r-\alpha (D_{1+\alpha}(\rho\|\sigma) )}
+d_{n,d} e^{-nr} .\,
\end{align}
where $(f)$ follows from the monotonicity of sandwitch R\'{e}nyi relative entropy.

\bibliographystyle{quantum}
\bibliography{references}

\end{document}